\documentclass[a4paper,UKenglish]{lipics-v2016}

\usepackage{bussproofs,enumerate,forest}
\usetikzlibrary{patterns}
\usepackage{amsthm,mik-makro,amsmath,graphicx}

\newcommand{\Rall}{G_{\forall}}
\newcommand{\Rone}{G_{1}}
\newcommand{\Rnonzero}{G_{>0}}

\newcommand{\npconp}{{\sc np}~$\cap$~{\normalfont co-}{\sc np}}

\usepackage{amsthm}
\usepackage{amsmath}
\usepackage{mik-makro}
\usepackage{graphicx}
\usepackage{xspace}
\usepackage[all]{xy}
\usepackage{float}
\usepackage{todonotes}
\usepackage[normalem]{ulem}
\usepackage{tikz}
\usepackage{tikz-qtree}

\usetikzlibrary{trees,automata,positioning}

\newcommand{\intuition}[1]{}

\newcommand{\Qall}{F_{\forall}}

\newcommand{\Qone}{F_{1}}
\newcommand{\Qnonzero}{F_{>0}}
\newcommand{\Qseed}{Q_{\mathrm{seed}}}

\newcommand{\ignore}[1]{}

\newcommand{\zero}{\mathsf {zero}}

\begin{document}
	
\bibliographystyle{plainurl}

\title{Emptiness of zero automata is decidable\footnote{The research of M. Bojańczyk is supported by the ERC grant LIPA under the Horizon 2020 framework. H. Gimbert and E. Kelmendi are supported by the French ANR project
    "Stoch-MC" and "LaBEX CPU" of Université de Bordeaux.}}
\titlerunning{Emptiness of zero automata is decidable} 

\author[1]{Mikołaj Bojańczyk}
\author[2]{Hugo Gimbert}
\author[2]{Edon Kelmendi}
\affil[1]{Institute of Informatics, University of Warsaw, Poland\\
  \texttt{bojan@mimuw.edu.pl}}
\affil[2]{LaBRI, Universit{\' e} de Bordeaux, CNRS, France\\
  \texttt{\{hugo.gimbert, edon.kelmendi\}@labri.fr}}

\authorrunning{M. Bojańczyk, H. Gimbert and E. Kelmendi}

\Copyright{Mikołaj Bojańczyk, Hugo Gimbert and Edon Kelmendi}


\EventEditors{John Q. Open and Joan R. Acces}
\EventNoEds{2}
\EventLongTitle{42nd Conference on Very Important Topics (CVIT 2016)}
\EventShortTitle{CVIT 2016}
\EventAcronym{CVIT}
\EventYear{2016}
\EventDate{December 24--27, 2016}
\EventLocation{Little Whinging, United Kingdom}
\EventLogo{}
\SeriesVolume{42}
\ArticleNo{23}

\maketitle
\begin{abstract}
Zero automata are a probabilistic extension of parity automata on infinite trees.
The satisfiability of a certain probabilistic variant of {\sc mso}, called {\sc tmso + zero},
reduces to the emptiness problem for zero automata.
We introduce a variant of zero automata called
nonzero automata. We prove
that for every zero automaton there is an equivalent nonzero automaton of quadratic size
and the emptiness problem of nonzero automata is decidable, with complexity \npconp.
These results imply that {\sc tmso + zero} has decidable satisfiability.
\end{abstract}
\section{Introduction}

In this paper, we prove that emptiness is decidable for two classes of automata, namely \emph{zero} and \emph{nonzero} automata. Zero automata were introduced as a tool for recognizing models of a probabilistic extension of MSO on infinite trees~\cite{thinzero}. Nonzero automata, introduced in this paper, are equivalent to zero automata,
but have simpler semantics.

Both zero and nonzero automata are probabilistic extensions of parity automata on infinite trees.
Here we focus on the case of binary trees.
The automaton performs a random walk on the infinite binary input tree: 
when the automaton is in a state $q$ on a node labelled with $a$,
it selects non-deterministically a transition $(q,a,r_0,r_1)$
and moves with equal probability $\frac{1}{2}$ either to the left node in state $r_0$
or to the right node in state $r_1$. 

The set of branches of the infinite binary tree is equipped with the uniform probability measure,
which is used to define the acceptance condition.
There are two variants of the acceptance condition, one for zero automata and one for nonzero automata
  
 A \emph{nonzero} automaton is equipped with a total order $\leq$ on its set of states $Q$ and three accepting subsets of states $\Qall,\Qone$ and $\Qnonzero$. 
 A run is accepting if: 
 \begin{itemize}
 \item[a)] on every branch the limsup state (i.e.~the maximal state seen infinitely often) is in $\Qall$,
 \item[b)]
with probability $1$ the limsup state is in  $\Qone$,
 \item[c)]
every time the run visits a state in $\Qnonzero$ there is nonzero probability
that all subsequent states are in  $\Qnonzero$. 
\end{itemize}
Condition (a) is the classical parity condition for tree automata and condition (b) is equivalent to the qualitative condition from~\cite{DBLP:journals/tocl/CarayolHS14}. 
Condition (c) seems to be new.
Conditions (a) and (b) are used to define the acceptance condition of zero automata as well,
the difference between zero and nonzero automata lies in condition (c).

The paper~\cite{thinzero}  introduced a variant of {\sc mso} on infinite trees with a probabilistic quantifier, 
called {\sc tmso}+$\zero$,  inspired by probabilistic {\sc mso} from~\cite{DBLP:conf/lfcs/MichalewskiM16}.
In the case where $\zero$ is the unary predicate which checks whether a set of branches has probability $0$,
the contribution of~\cite{thinzero} was a proof that for every formula of this logic one can compute a zero automaton which accepts the same trees. The logic is powerful enough to formulate properties like "every node in the tree has a descendant node labelled with $b$ and the set of branches with infinitely many $b$ has probability $0$".
As argued in~\cite{thinzero}, the motivation for this logic is twofold. First, it  extends various probabilistic logics known in the literature, e.g.~qualitative probabilistic {\sc ctl*}~\cite{LS1982}, or  qualitative probabilistic  {\sc ctl*}  extended with $\omega$-regular path properties~\cite{DBLP:conf/icalp/BrazdilFK08}. Second, the logic, although less general that {\sc mso}, represents a robust class of languages of infinite trees that goes beyond classical {\sc mso}, and thus falls under the scope of the programme of searching for decidable extensions of {\sc mso}.



 The emptiness problem for zero automata was not solved in~\cite{thinzero}, thus leaving open the logic's decidability. A step toward an emptiness algorithm was made in~\cite{DBLP:journals/corr/MichalewskiMB16}, where it was shown that for subzero automata --  the special case of zero automata where only conditions (a) and (b) are used -- one can decide if the recognised language contains a regular tree. In this paper we prove that zero and nonzero automata have decidable emptiness, and therefore also the logic from~\cite{thinzero} has decidable satisfiability.  

The main results of this paper are:
\begin{enumerate}
\item[i)]
For every zero automaton there is an equivalent nonzero automaton of quadratic size.
\item[ii)] 
A nonzero automaton with $\Qall=Q$ is nonempty if and only if its language contains a regular tree of size $|Q|$.
This is decidable in polynomial time.
\item[iii)]
The emptiness problem of nonzero automata is in \npconp.
\end{enumerate}

To prove iii) we provide a reduction of the emptiness problem to the computation of the winner of a parity game
called the \emph{jumping game}.
For that we rely on ii): the states of the jumping game are regular runs of a nonzero automaton where $\Qall=Q$.
According to i) the emptiness problem for zero automata is in \npconp as well.

The plan of the paper is as follows. In Section~\ref{logic_section} we introduce zero and nonzero automata and
state our main result  iii) (Theorem~\ref{thm:emptiness}).
In Section~\ref{sec:zerononzero} we show
i) (Lemma~\ref{lem:same-models}). In Section~\ref{sec:proba} we focus on the special case where $Q=\Qall$ and show ii) (Theorem~\ref{theo:probanp}). In Section~\ref{sec:jumping} we introduce  jumping games
and combine the previous results to provide a proof of iii).
%

\newcommand{\probquant}[2]{\underset{#1}{\mathsf P}(#2)}

\section{Zero and nonzero automata} \label{logic_section}
This section introduces trees and nonzero and zero automata.

\paragraph*{Trees, branches and subtrees.} The automata of this paper describe properties of infinite binary labelled trees. A node in a  tree is a sequence in $\set{0,1}^*$. A \emph{tree} over an alphabet $\Sigma$ is a function $t : \set{0,1}^* \to \Sigma$.  We use standard terminology for trees: node, root, left child, right child,  leaf, ancestor and descendant. A \emph{branch} is a sequence  in $\{0,1\}^\omega$, viewed as an infinite sequence of left or right turns.  A branch \emph{visits} a node if the node is a prefix of the branch.

A \emph{subtree} is a non-empty and ancestor-closed set of nodes.
A subtree is \emph{leaf-free} if each of its nodes has at least one child in the subtree.
A branch of a subtree is a branch which visits only nodes of the subtree.


\paragraph*{Probability measure over branches.} 
We use the \emph{coin-flipping} measure on $\{0,1\}^\omega$:
each bit is chosen independently at random, with $0$ and $1$ having equal probability,
and every Borel subset of $\{0,1\}^\omega$ is measurable.
The probability of a subtree is the probability of the set of branches of the subtree.
%
%
The inner regularity of the coin-flipping measure (see e.g. \cite[Theorem 17.10]{kechris1995classical}) implies:
%
\begin{lemma}\label{lem:approx}
The probability of a measurable set $E$ 
is the supremum of the probabilities
of the subtrees whose every branch belongs to $E$.
\end{lemma}

\paragraph*{Nonzero automata}
Intuitively,
a nonzero automaton is a nondeterministic parity tree automaton which has the
extra ability to check whether the set of branches satisfying the parity condition has zero or nonzero probability.

\begin{definition}\label{def:zero automata}
	The syntax of a nonzero automaton is a tuple
	\begin{align*}
		\underbrace{Q}_{\text{states}} \qquad \underbrace{\Sigma}_{\text{input alphabet}}  \qquad \underbrace{\delta \subseteq Q \times \Sigma \times Q^2}_{\text{transitions}},
	\end{align*}
	with all components finite, together with a total order $\leq$ on $Q$ and three subsets 
	\begin{align*}
		\Qall, \Qone, \Qnonzero  \subseteq Q\enspace.
	\end{align*}
\end{definition} 

A \emph{run} of the automaton on an input tree $t:\{0,1\}^*\to \Sigma$
 is an infinite binary tree $r:\{0,1\}^*\to Q$
 whose root is labelled by the maximal state of $Q$,
also called the \emph{initial} state
and
 which is consistent with the transition relation
in the usual sense,
 i.e.
$ \forall v \in\{0,1\}^*, (r(v),t(v),r(v0),r(v1))\in\Delta$.
Define the \emph{limsup} of a branch of the run to be the maximal state that appears infinitely often on the branch.

The run is \emph{accepting} if it is surely, almost-surely and nonzero accepting:
\begin{itemize}
	\item \label{sem:all} {\bf surely accepting:} every branch has limsup  in   $\Qall$.
	\item  \label{sem:zero} {\bf almost-surely accepting:} 
	the set of branches with limsup in $\Qone$ has probability $1$.
	\item \label{sem:nonzero} {\bf nonzero accepting:} for every node $v$ with state in $\Qnonzero$,
	the set of branches which visit $v$ and visit only $\Qnonzero$-labelled nodes below $v$
	has nonzero probability.
\end{itemize}
\paragraph*{The emptiness problem}

The emptiness problem asks whether an automaton has an accepting run.
Our main result:
\begin{theorem}\label{thm:emptiness}
	The emptiness problem of nonzero automata is decidable in \npconp.
\end{theorem}
\begin{proof}
This is a corollary of a series of intermediary results.
 In section~\ref{sec:proba} we focus on the special case where 
$\Qall=Q$ and provide an polynomial time algorithm to decide emptiness in this special case (Theorem~\ref{theo:probanp}).
In section~\ref{sec:jumping} we reduce the emptiness problem for nonzero automata to the computation of the winner in a parity game called the \emph{jumping game} (Lemma~\ref{lem:jumping-game-iff}) and give an \npconp algorithm to compute the winner of the jumping game (Lemma~\ref{lem:jumping-game-decide}).
\end{proof}

\paragraph*{Zero automata}

Nonzero automata are a variant of \emph{zero automata} introduced in~\cite{thinzero}.
A zero automaton differs slightly from a nonzero automaton in that it uses a notion of ``seed state'' for the nonzero acceptance condition. On top of $\Qall,\Qone$ and $\Qnonzero$ there is a subset $\Qseed \subseteq Q$. A run is accepting if it is surely, almost-surely and \emph{zero} accepting:
\begin{itemize}
\item{\bf  zero accepting:}
for every node $v$ with state $q\in\Qseed$,
there is nonzero probability that the run
visits only states $\leq q$ below $v$
and has limsup in $\Qnonzero$.
\end{itemize}

In the next section,
we show that every zero automaton can be transformed in an equivalent nonzero automaton of quadratic size (Lemma~\ref{lem:same-models}). Combined with
Theorem~\ref{thm:emptiness},
\begin{corollary}\label{thm:emptinesszero}
	The emptiness problem of zero automata is in \npconp.
\end{corollary}
According to~\cite{thinzero},
this implies that {\sc tmso }+ $\zero$
has decidable satisfiability 
when $\zero$ is the unary predicate checking that a set of branches has probability $0$.
 
 \paragraph*{An example: the dense but not very dense language}
A tree over alphabet $\set{a,b}$ is  \emph{dense but not very dense} if:
\begin{enumerate}
	\item every node has a descendant with label $a$; and
	\item there is zero probability that a branch visit infinitely many nodes with letter $a$.
\end{enumerate}
This language is non-empty,
contains no regular tree
and is recognised by a nonzero automaton.
This automaton has three states, totally ordered as follows:
\begin{align*}
  	  \underbrace{s}_{\text{searching for $a$}} < \qquad  \underbrace{n}_{\text{not searching for $a$}} \qquad < \qquad 	\underbrace{f}_{\text{just found $a$}} \enspace.
\end{align*}
The automaton begins in state $f$ in the root.
When the automaton reads a node with label $b$, then it sends $s$ to some child and $n$ to the other child,
regardless of its current state.
Choosing which child gets $s$ and which child gets $n$ is the only source of nondeterminism in this automaton.
When the automaton sees  letter $a$,  it sends $f$ to both children regardless of its current state.
The acceptance condition is:
 \begin{align*}
 	\Qall = \set{n,f} \qquad \Qone = \set{n} \qquad \Qnonzero = \emptyset\enspace.
 \end{align*}


\section{From zero to nonzero automata\label{sec:zerononzero}}

In this section we show that nonzero automata are as expressive as zero automata.

\begin{lemma}\label{lem:same-models}
	For every zero automaton one can compute a nonzero automaton of quadratic size which accepts the same trees.
\end{lemma}

The rest of the section is dedicated to the proof of Lemma~\ref{lem:same-models},
which is a direct corollary of Lemma~\ref{lem:projnonzero} and Lemma~\ref{lem:zerotononzero} below.

Without loss of generality, we assume that in every
\emph{zero} automaton $
\Qnonzero \subseteq \Qone \subseteq \Qall$.
Changing $\Qone$ for $\Qone \cap \Qall$ and
$\Qnonzero$ for  $\Qnonzero \cap \Qone$
does not modify the set of accepting runs of a zero automaton,
since all branches should have limsup in $\Qall$
and if the limsup is equal with nonzero probability to some $q\in\Qnonzero$ then necessarilly $q\in \Qone$.
By contrast, for \emph{nonzero} automata there is no obvious reason for the same remark to hold.



%

We make use of an intermediary acceptance condition. Let $r$ be a run. We say that a path from a node $v$ to a node $w$
is \emph{seed-consistent} if whenever the path visits a seed state $s$, subsequent states are $\leq s$.
\begin{itemize}
\item
{\bf Strong zero acceptance condition:}
for every node $v$ labelled by a seed state,
there is a seed-consistent path from $v$ to a
strict descendant $w$ of $v$ such that the state $r(w)$ of $w$ is in $\Qnonzero$
and there is nonzero probability that the run
\begin{itemize}
\item
visits only states $\leq r(w)$ below $w$,
\item
has limsup $r(w)$,
\item
in case $ r(w)\not\in\Qseed$, visits no seed state below $w$,
\item
in case $ r(w)\in\Qseed$, visits no seed state other than $ r(w)$ below $w$.
\end{itemize}
\end{itemize}

Actually, the strong zero and zero acceptance conditions coincide (proof in appendix):
\begin{lemma}\label{lem:strong}
A run is zero accepting if and only if it is strongly zero accepting.
\end{lemma}

\paragraph*{Construction of the nonzero automaton}
Intuitively, every \emph{zero} automaton can be simulated by a \emph{nonzero} automaton which guesses on the fly a run of the zero automaton and checks simultaneously that the guessed run is strongly zero accepting.
Whenever the automaton visits a node $v$ with a seed state
then it enters
in the next step a \emph{path-finding state} and guesses a seed-consistent path to a node $w$ which is a witness of the strong zero condition. 
Once on the node $w$ the automaton enters  a \emph{subtree-guessing} state and starts guessing a leaf-free subtree of the run, whose nodes are labelled by states $\leq  r(w)$, whose branches have limsup $ r(w)$ and which has nonzero probability.

There are some verifications to do in order to certify that the guessed run is strongly zero accepting.
The surely accepting condition is used to prevent the automaton to stay forever in the path-finding mode and also to check that every  branch of the subtree has limsup $ r(w)$.
The nonzero condition is used to check that the subtree has nonzero probability.
To perform these verifications, the nonzero automaton stores some data in its control state. In path-finding mode the automaton records the smallest seed state seen so far in order to check  on-the-fly that the path from $v$ to $w$ is seed-consistent.
In subtree-guessing mode the automaton keeps track of the state $ r(w)$.

The set of states of this automaton is denoted $R$,
every state in $R$ has as a first component a
control state $Q$ of the zero automaton.  Precisely, $R$ is the union of three sets:
\begin{itemize}
\item
{\bf normal states:}
$Q$
\item
{\bf path-finding states:}
$\{(q,s) \mid q \in Q, s\in\Qseed, q\leq s\}$,
\item
{\bf subtree-guessing states:}
 $\{ (q,f,*) \mid q \in Q, f \in \Qnonzero, q \leq f, (q\not \in\Qseed \lor q=f)\}$.
\end{itemize}

We equip $R$ with any order $\prec$ such that
\begin{itemize}
\item
the projection on the first component  $\Pi_1: (R,\prec)\to (Q,<)$
 is monotonic,
\item
$(q,s) \prec q$ for every $q\in Q$ and $s\in\Qseed$ with $q\leq s$.
\end{itemize}
The zero, almost-surely and surely
accepting conditions
are defined respectively as:
\begin{align*} 
& \Rnonzero = \text{ the set of subtree-guessing states,}\\
& \Rone = \Qone \cup  \{(f,f,*) \mid f \in \Qnonzero\},\\
& \Rall = \Qall \cup  \{(f,f,*) \mid f \in \Qnonzero\}
\enspace.
\end{align*} 

The transitions of the automaton can be informally described as follows.
The nonzero automaton guesses on the fly a run $\rho:\{0,1\}^*\to Q$ of the zero automaton by storing the value
of $ \rho(v)$ as the first component of its own control state on the node $v$.
The nonzero automaton stays in the set of normal states as long as the run does not enter a seed state. On a node $v$ labelled by $s\in\Qseed$,
the nonzero automaton  starts looking for a path to a descendant node $w$ that satisfies the  strong zero condition.
For that in the next step the automaton enters either a path-finding or a subtree-guessing state.
While in a path-finding state, the automaton guesses on the fly a seed-consistent path. Whenever the run is in a nonzero state $f\in\Qnonzero$ the nonzero automaton can enter the subtree-guessing state $(f,f,*)$, or not.
While in subtree-guessing mode the second component is constant, and the automaton control state is of type
$(q,f,*)$ with $q\leq f$ and $q\not\in\Qseed$ unless $q=f\in\Qseed$.
From a subtree-guessing state the automaton may switch back any time to a normal state.

Formally, for every transition $q\to r_0,r_1$ of the zero automaton,
there is a transition
\[
q'\to r_0',r_1'
\]
in the nonzero automaton if
the first component of $q'$ is $q$ and 
\[
r'_0=
\begin{cases}
r_0 & \text{ whenever $q'$ is not path-finding}\\
(r_0,r_0,*) & \text{ whenever }
\begin{cases}
\text{$q\in \Qseed,q'=q$
and $r_0\in \Qnonzero$ and $r_0\leq q$}\\
\text{or  $q'=(q,s)$ 
and $r_0\in \Qnonzero$ and $r_0\leq s$,}
\end{cases}\\
(r_0,f,*) & \text{ whenever $q'=(q,f,*)$ and $r_0\leq f$
and $(r_0\not\in \Qseed \lor r_0 = f )$.}\\
\end{cases}
\]
The possible values of $r'_1$ are symmetric.
There are also \emph{left path-finding transitions}: for every seed states $s,s'\in\Qseed$ such that $q\leq s$ and $r_0\leq s$
there are transitions
\[
q' \to (r_0,s'),r_1
\text{ where } 
q'=
\begin{cases}
q \text{ or } (q,q) \text{ if } q =s\\
(q,s) \text{ otherwise }
\end{cases}
\text{ and }
s' =
\begin{cases}
s & \text{ if $r_0\not \in\Qseed$}\\
r_0 & \text{ if $r_0 \in\Qseed$}.
\end{cases}
\]
There may also be a symmetric  \emph{right path-finding transition}
$
(q,s)\to r_0,(r_1,s')
$
when the symmetric conditions hold.


The next two lemmas relate the accepting runs of the zero and the nonzero automata,
their proofs can be found in the appendix.

\begin{lemma}\label{lem:projnonzero}
Let $d:\{0,1\}^*\to R$ be an accepting run of the nonzero automaton.
Then its projection $r=\Pi_1(d)$ on the first component 
is an accepting run of the zero automaton.
\end{lemma}

\begin{lemma}\label{lem:zerotononzero}
If the zero automaton has an accepting run $r:\{0,1\}^*\to Q$
then the nonzero automaton has an accepting run $d:\{0,1\}^*\to R$ such that $r=\Pi_1(d)$.
\end{lemma}
\section{Emptiness of $\Qall$-trivial automata is in NP}
\label{sec:proba}
 A run of a nonzero automaton needs to satisfy simultaneously three conditions, which correspond to the accepting sets $\Qall,\Qone,\Qnonzero$. For a subset 
\begin{align*}
  I \subseteq \set{\Qall,\Qone,\Qnonzero }
\end{align*}
define $I$-automata to be the special case of nonzero automata where only the acceptance conditions corresponding to $I$ need to be satisfied. These are indeed special cases: ignoring $\Qnonzero$ can be achieved by making it empty, ignoring $\Qone$ can be achieved by making it equal to $\Qall$, and ignoring $\Qall$ can be achieved by making it equal to all states $Q$.

\paragraph*{Generalising parity automata, with standard and qualitative semantics}

A $\set{\Qall}$-automaton is a parity automaton. Thus solving emptiness for nonzero automata is at least as hard as emptiness for parity automata on trees, which is polynomial time equivalent to solving parity games,
in {\sc np }$\cap$ co{\sc np} or in quasi-polynomial time~\cite{calude}.

A $\set{\Qone}$-automaton is the same as a parity automaton with qualitative semantics as introduced in~\cite{DBLP:journals/tocl/CarayolHS14}.  Emptiness for such automata can be solved in polynomial time using standard linear programming algorithms for Markov decision processes.

%



\paragraph*{Subzero automata}
A $\set{\Qone,\Qall}$-automaton is the same as a \emph{subzero} automaton as considered in~\cite{DBLP:journals/corr/MichalewskiMB16}. In~\cite{DBLP:journals/corr/MichalewskiMB16}, it was shown how to decide if a subzero automaton accepts some regular tree.   Since some subzero automata are nonempty but accept no regular trees, see e.g.~the example in \cite{thinzero}, the result from~\cite{DBLP:journals/corr/MichalewskiMB16} does not solve nonemptiness for subzero automata. 

\paragraph*{$\Qall$-trivial automata}
In a $\set{\Qone,\Qnonzero}$-automaton, the surely accepting condition is trivial, i.e. $\Qall=Q$.
We call such automata $\Qall$-trivial.
The 
acceptance of a run of a $\Qall$-trivial automaton
depends only on the probability measure on $Q^\omega$ induced by the run,
individual branches do not matter.

\begin{definition}[Positional run]
A run is \emph{positional} if
whenever the states of two nodes coincide
then the states of their left children coincide
and
the states of their right children coincide.
\end{definition}

\begin{theorem}\label{theo:probanp}
If a $\Qall$-trivial automaton
has an accepting run, then it has a positional accepting run.
Emptiness of $\Qall$-trivial automata can be decided in polynomial time.
\end{theorem}

The proof of this theorem relies
on the notion of acceptance witnesses.

\begin{definition}[Transition graph and acceptance  witness]
Let $D$ be a set of transitions.

The transition graph of $D$, denoted $G_D$, is the directed graph
whose vertices are all states appearing in one of the transitions in $D$,
denoted $Q_D$,
and
whose edges are induced by the transitions in $D$: for every
$(q,a,l,r)\in D$ both $(q,l)$ and $(q,r)$ are edges of $G_D$.

The set $D$
is an \emph{acceptance witness} if
it satisfies the four following conditions:
\begin{itemize}
\item[i)]
$Q_D$ contains the initial state of the automaton and $G_D$ has no dead-end,
\item[ii)]
the maximum of every bottom strongly connected component (BSCC) of $G_D$
is in $\Qone$,
\item[iii)]
every BSCC of $G_D$ is either contained in $\Qnonzero$
or does not intersect $\Qnonzero$,
\item[iv)]
from every state in $\Qnonzero\cap Q_D$ there 
is a path in $\Qnonzero\cap Q_D$ to a BSCC contained in $\Qnonzero$.
\end{itemize}
\end{definition}

\begin{lemma}\label{lem:bscc}
If a $\Qall$-trivial automaton has an acceptance witness,
it has a positional accepting run.
\end{lemma}
\begin{proof}
The proof is by induction on $N_D= |D|-|Q_D|$.
Since $G_D$ has no dead-end, every state in $Q_D$ is the source
of a transition in $D$ thus  $N_D\geq 0$.

If $N_D=0$ then for every state $q\in Q_D$
there is a unique transition $\delta_q=(q,a_q,l_q,r_q)$.
Let $\rho$ be the positional run whose root has the initial state
and every node with vertex $q\in Q_D$ has children $l_q$ and $r_q$,
which is well-defined according to property i).
We show that $\rho$ is an accepting run.
The graph $G_D$ can be seen as a Markov chain,
with probability either $1$ or $\frac{1}{2}$ on every edge, depending on the outdegree.
The probability measure on $Q_D^\omega$
produced by the random walk on $\rho$ coincide with the probability measure
on $Q_D^\omega$ produced by this finite Markov chain:
indeed both measures coincide on finite cylinders $q_0\cdots q_nQ_D^\omega$.
Basic theory of finite homogenous Markov chain implies
that almost-surely every branch of the run ends up in one of the BSCCs
of $G_D$ and visits all its states infinitely often.
Thus property ii) ensures that the run $\rho$ is almost-surely accepting.
Properties iii) and iv) guarantee that
the run is moreover nonzero-accepting.

Assume now that $N_D>0$.
We show that there is a strictly smaller acceptance witness $D'\subsetneq D$.
Let $q\in Q_D$ which is the source of several transitions in $D$,
then $D'$ is obtained by removing from $D$ all these transitions except one.
To choose which transition $\delta$ to keep,
we pick up the shortest path $q=q_0\ldots q_n$ in $G_D$ 
of length $\geq 1$ which leads to the maximal state of one of the BSCCs of $G_D$.
Moreover if $q\in\Qnonzero$ we require the whole path to stay in $\Qnonzero$.
By definition of $G_D$ there is at least one transition in $ D$
whose origin is $q$ and one of the two successors is $q_1$.
To get $D'$ we delete all other transitions with source $q$ from $D$.

Clearly property i) is preserved by this operation.
To address properties ii)-iv),
we show that
every BSCC $B'$ of $G_{D'}$ is either a BSCC of $G_D$
or contained in the BSCC $B$ of $G_D$ whose maximum is
$q_n$, in which case $\max B=\max B'=q_n$.
There are two cases.
If $B'$ does not contain $q_n$ then it does not contain $q$ either
(because $q=q_0\ldots q_n$ is still a path in $G_{D'}$).
Since the only difference between $G_D$ and $G_{D'}$ are the outgoing transitions
from $q$ then $B'$ is actually a BSCC of $G_D$.
If $B'$ contains $q_n$ then $B'\subseteq B$
(because there are less edges in $G_{D'}$ than in $G_D$) and since 
$q_n=\max B$ then $\max B =\max B'$.

As a consequence property ii) and iii) are preserved.
And property iv) is preserved as well: in case $q\not\in\Qnonzero$
then there is nothing to prove and in case $q\in \Qnonzero$
then $q=q_0\ldots q_n$ is still a path in $G_{D'}$,
with all vertices in $\Qnonzero$. Moreover the set of vertices from which $q_n$ is 
accessible is the same in $G_D$ and $G_{D'}$
thus $q_n$ is in a BSCC of $G_{D'}$.
\end{proof}

A strong version of the converse implication of Lemma~\ref{lem:bscc} holds:
\begin{lemma}\label{lem:witness}
If a $\Qall$-trivial automaton has an accepting run,
it has an acceptance witness.
\end{lemma}
\begin{proof}
We fix an accepting run $\rho$
on some input tree $t$.
To extract an acceptance witness from $\rho$,
we make use of the notion of end-component
introduced in~\cite{deAlfaro:1997}.
\begin{definition}[End-component]
The \emph{transition of a node} $v$ 
is  $d(v)=(\rho(v),t(v),\rho(v0),\rho(v1))$.
For every branch $b$,
we denote $\Delta^\infty(b)$
the set of transitions labelling infinitely many nodes of the branch.
For every subset $D\subseteq \Delta$ we denote $B_D$ the set of branches $b$
such that $\Delta^\infty(b)= D$.
A set of transitions $D\subseteq \Delta$ is an \emph{end-component}
of the run if $B_D$ has nonzero probability.
\end{definition}

Call a branch $b$ \emph{even} if for every transition
$\delta=(q,a,l,r)\in\Delta^\infty(b)$,
not only the state $q$ but also the states
$l$ and $\rho$ appear infinitely often on the branch in the run $\rho$.
Almost-surely every branch is even,
because each time a branch visits
a node with transition $\delta$ it proceeds left or right with equal probability
$\frac{1}{2}$.
As a consequence,
\begin{lemma}\label{lem:endc}
Let $D$ be an end-component of the run.
Then the transition graph of $D$ has no dead-end,
is strongly connected and its maximal state is in $\Qone$.
\end{lemma}
\begin{proof}
Denote $G_D$ the transition graph of $D$, with states $Q_D$.
Since $D$ is an end-component then $B_{D}$ has non-zero probability,
and since almost every branch is even
then $B_D$ contains at least one even branch $b$.
The set of states appearing infinitely often on $b$ is exactly $Q_D$.
By removing a prefix long enough of $b$ so that
only states in $Q_D$ occur on the remaining suffix then one obtains 
a path in $G_D$ which visits every state in $Q_D$ infinitely often.
Thus $G_D$ has no dead-end and is strongly connected.
Moreover every even branch in $B_D$ has limsup
$\max Q_D$ and since
the run is almost-surely accepting then $\max Q_D\in \Qone$.
\end{proof}

Let $\mathcal{D}$ be the collection of all end-components of the run $\rho$.
We define two subsets of $\mathcal{D}$,
denoted respectively $\mathcal{D}_0$ and $\mathcal{D}_1$,
which collect  the end-components
whose states are respectively included in $\Qnonzero$
and disjoint from $\Qnonzero$.
Let $D_0\subseteq \Delta$ (resp. $D_1\subseteq \Delta$) be the union of all end-components
in $\mathcal{D}_0$ (resp. in $\mathcal{D}_1$).
These transitions 
are easy to reach:

\begin{lemma}\label{lem:path}
Every node $v$ 
has a descendant $w$ whose transition belongs to $D_0\cup D_1$.
Moreover if the state of $v$ is in $\Qnonzero$
then $w$ can be chosen such that the path $v$ to $w$ is labelled by $\Qnonzero$
and the transition is in $D_0$.
\end{lemma}
\begin{proof}
Let $v$ be a node
and $S_v$ the set of branches which visit $v$
and, in case $v$ is labelled by $\Qnonzero$, visit
only $\Qnonzero$-labelled nodes below $v$.
Since the run is accepting then $S_v$ has positive probability.
By definition of end-components,
almost-every branch is in $\bigcup_{D\in\mathcal{D} } B_D$.
Thus there exists an end-component $D$ such that 
 $B_D \cap S_v$ has positive probability.
 As a consequence, $v$ has a descendant  $w$
 whose transition is in $D$.
 %
%
%
%
Since
almost-every branch is even
and $B_D\cap S_v$ has positive probability then
there is at least one branch in $B_D\cap S_v$ which visits
infinitely often all states appearing in $Q_D$.
In case $v$ is labelled by $\Qnonzero$,
this implies that $Q_D\subseteq \Qnonzero$ thus $D\in\mathcal{D}_0$,
and terminates the proof of the second statement.
In case $v$ has no descendant labelled by $\Qnonzero$
this implies that $Q_D\cap \Qnonzero=\emptyset$ thus $D\in\mathcal{D}_1$,
and the first statement holds in this case.
In the remaining case, $v$ has a descendant $v'$ labelled with $\Qnonzero$,
which itself has a descendant $w$ whose transition belongs to some $D\in \mathcal{D}_0$,
thus the first statement holds for $v$.
\end{proof}
We terminate the proof of Lemma~\ref{lem:witness}.
Let $G_0$ (resp. $G_1$) the transition graph of $D_0$ (resp. $D_1$)
and denote $Q_0$ (resp $Q_1$) the set of 
states of $G_0$ (resp. $G_1$).

Let $D$ be the set of all transitions appearing in the run.
According to Lemma~\ref{lem:path},
in the transition graph $G_D$,
$Q_0\cup Q_1$ is accessible from every state $q \in Q_D$
and moreover $Q_0$ is accessible from every state $q \in Q_D \cap \Qnonzero$
following a path in $Q_D \cap \Qnonzero$.

We say that an edge $(q,r)$ of $G_D$ is \emph{progressive} if $q\not\in Q_0\cup Q_1$
and either ($q\in\Qnonzero$
and $r\in\Qnonzero$ and $(q,r)$ decrements the distance to $Q_0$ in $G_D$)
or  ($q\not\in\Qnonzero$ and $(q,r)$ decrements the distance to $Q_0\cup Q_1$ in $G_D$).
Every state in $Q_D\setminus (Q_0\cup Q_1)$ is the source of at least one progressive edge.

We denote $D_+$ the union of $D_0$ and $D_1$ plus all the transitions $\delta=(q,a,r_0,r_1)\in D$
such that either $(q,r_0)$ or $(q,r_1)$ is progressive.
Then $D_+$ has all four properties of Lemma~\ref{lem:bscc}.
Denote $G_+$ the transition graph associated to $D_+$.
Property i) holds because
every state in $Q_D$, including the initial state, is either in 
$Q_0\cup Q_1$ or is the source of a progressive edge.

Remark that the BSCCs of $G_+$ are exactly the BSCCs of $G_0$ and $G_1$.
Since both $G_0$ and $G_1$ are unions of strongly connected graphs,
they are equal to the union of their BSCCs.
The BSCCs of $G_0$ and $G_1$ are still BSCCs in $G_+$
because no edges are added inside them
(progressive edges have their source outside $G_0$ and $G_1$).
Following the progressive edges leads to $G_0$ or $G_1$
from every state in $G_+$,
thus there are no other BSCCs in $G_+$.

This implies property ii)  because,
according to Lemma~\ref{lem:endc},
both graphs $G_0$ and $G_1$ are the union of strongly connected graphs
whose maximal states are in $\Qone$.
This also implies property iii) since $Q_0\subseteq \Qnonzero$
and $Q_1 \cap \Qnonzero=\emptyset$.
Property iv) is obvious for states in $Q_0$
because $Q_0$ is a union of BSCCs included in $\Qnonzero$.
Property iv) holds as well for states in $(Q_D\cap \Qnonzero) \setminus Q_0$,
the path to $Q_0$ is obtained following the progressive edges in 
$\Qnonzero\times \Qnonzero$.
\end{proof}

\begin{proof}[Proof of Theorem~\ref{theo:probanp}]
According to Lemma~\ref{lem:witness} and Lemma~\ref{lem:bscc},
non-emptiness of a $\Qall$-trivial automaton is equivalent
to the existence
of an acceptance witness,
which implies the existence of a positional accepting run.
Guessing a subset of transitions and checking it is an acceptance witness can be done in 
non-deterministic polynomial time.

Actually it is possible to check the existence of an acceptance witness in polynomial time. Using standard algorithms for Markov decision processes, one can compute the set 
$R_0$ (resp. $R_1$) of states $q$ such that there exists an almost-surely accepting run with root state $q$
and whose states are labelled by  $\Qnonzero$ (resp. by  $Q\setminus \Qnonzero$) (see~\cite[Corollary 18]{DBLP:journals/tocl/CarayolHS14} for more details).

We transform the $\Qall$-trivial automaton $\mathcal{A}$
into another $\Qall$-trivial automaton $\mathcal{A}'$
as follows. In $\mathcal{A}'$ every state $q$ in 
$R_0\cup R_1$ is turned into an absorbing state:
for every letter $a$ there is a transition $(q,a,q,q)$ and no other transition with source $q$.
Moreover we change the almost-sure condition and set it equal to $R_0\cup R_1$. The positive condition is not modified.

We claim that $\mathcal{A}$ has an accepting run if and only if $\mathcal{A}'$ has one.
Assume $\mathcal{A}$ has an accepting run.
Then it has an acceptance witness $D$.
According to ii), 
every BSCC $B$ of $D$ is included either in $R_0$ (if $B\subseteq \Qnonzero$) or in $R_1$ (if $B\cap \Qnonzero=\emptyset$). Thus $D$ can be turned into an  acceptance witness of $\mathcal{A}'$ by exchanging any transition $(q,a,l,r)$ with $q\in R_0\cup R_1$ into the absorbing transition $(q,a,q,q)$.
Conversely, assume 
 $\mathcal{A}'$ has an accepting run $\rho'$.
 Then by definition of $R_0$ and $R_1$
 every state $q\in R_0 \cup R_1$ is an acceptance witness of some almost-surely accepting run $\rho_q$ with root $q$ and all nodes in $\Qnonzero$ or out of $\Qnonzero$.
 Then we can build an accepting run of $\mathcal{A}$
 by modifying $\rho'$ as follows: for every node labelled by $q\in R_0\cup R_1$ with no ancestor labelled
by $R_0\cup R_1$ we replace the subtree by $\rho_q$.
Since almost-surely every path reaches $R_0\cup R_1$
then the new run is almost-surely and positively accepting.

The criteria for $D$ to be an acceptance witness of $\mathcal{A'}$ are simple:
$Q_D$ should contain the initial state and moreover:
\begin{itemize}
\item[a)]
$G_D$ has no dead-end,
\item[b)]
from every state in $Q_D$ there is a path in $Q_D$
to $R_1 \cup R_0$.
\item[c)]
from every state in $\Qnonzero\cap Q_D$ there 
is a path in $\Qnonzero\cap Q_D$ to $R_0$.
\end{itemize}
Notice that properties a), b) and c) are closed by union:
if both $D_1$ and $D_2$ have these three properties then $D_1\cup D_2$ as well.
And the largest set of transitions $D_{\max}$ with properties a) b) and c) is easy to compute in polynomial time:
start with $D_{\max}$ equal to all transitions and as long as possible remove:
\begin{itemize}
\item
any transition leading to a dead-end,
\item
 all transitions inside a BSSC disjoint from $R_0$ and $R_1$,
 \item
 all transitions $(q,a,q_0,q_1)$ such that
  $q\in\Qnonzero$
 and $R_0$ is not reachable from $q$ by a path in $\Qnonzero\cap Q_{D_{\max}}$.
 \end{itemize}
An invariant of this process is that all transitions of any acceptance witness are preserved. 
Finally, $\mathcal{A'}$ has an accepting run if and only if $Q_{D_{\max}}$ contains the initial state.
\end{proof}

\section{Emptiness of nonzero automata is in \npconp}
\label{sec:jumping}

In this section we show how to decide the emptiness of nonzero automata.
The main ingredient are jumping games.

Call  a run \emph{$\set{\Qone,\Qnonzero}$-accepting}
if it satisfies the almost-surely and the nonzero acceptance condition, but it does not necessarily satisfy the
surely accepting condition, and the root may not be labelled by the initial state either.

\paragraph*{The jumping game.}

For a run $\rho$, define its \emph{profile} $\Pi$ to be following set of state pairs:
\begin{multline*}
\Pi=\{(q,m) : \mbox{some non-root node in $\rho$ has state $q$}\\
\mbox{ and $m$ is the maximal state of its strict ancestors}\}\enspace.
\end{multline*}

The \emph{jumping game}  is a parity game
played by two players, \emph{Automaton} and \emph{Pathfinder}.  
Positions of  Automaton are states of the automaton and positions of  Pathfinder are profiles of  $\set{\Qone,\Qnonzero}$-accepting runs. The game is an edge-labelled parity game, i.e.~the priorities are written on the edges. The edges originating in Automaton positions are of the form
\begin{align*}
  q \stackrel q \to \Pi \qquad \mbox{such that $\Pi$ is the profile of some $\set{\Qone,\Qnonzero}$-accepting run with root state $q$.}
\end{align*}
The edges originating in Pathfinder positions are of the form
\begin{align*}
 \Pi \stackrel m \to q \qquad \mbox{such that $(q,m) \in \Pi$}.
\end{align*}
We say that  Automaton wins the jumping game if he has a winning strategy from the position which is the initial state of the automaton.  If the play ever reaches a dead-end, i.e. a state which is not the root of any $\set{\Qone,\Qnonzero}$-accepting run,  then the game is over and Automaton loses. Otherwise Automaton wins iff
the limsup of the states is in $\Qall$.

Lemmas~\ref{lem:jumping-game-iff} and~\ref{lem:jumping-game-decide} below establish that nonemptiness
of a nonzero automaton is equivalent to Automaton winning the jumping game, and this can be decided in  {\sc np}.

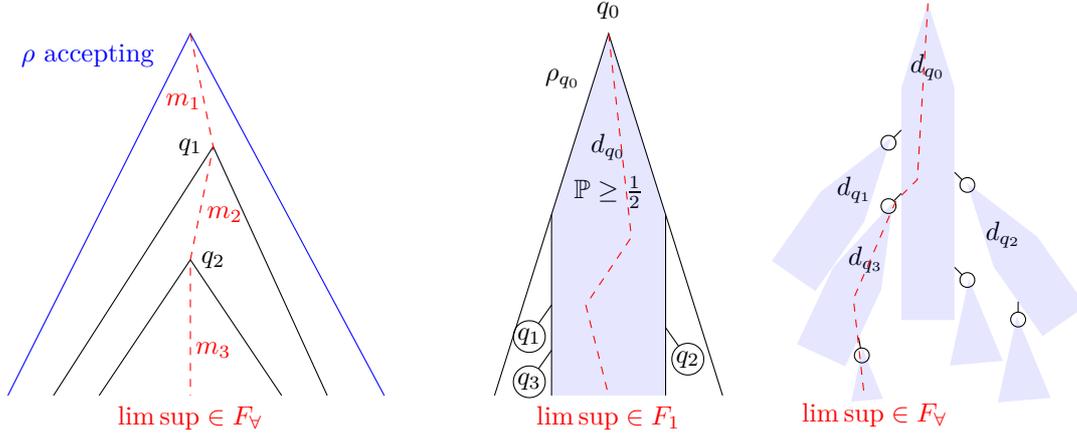
\begin{figure}[ht]
\begin{tikzpicture}
\begin{scope}[scale=.3]
\node[blue] at (5.5,14) { $\rho$ accepting};
\draw[blue] (2,-1) -- (10,15) -- (18.5,-1) ;
\node[red] at (9.7,12) {$m_1$};
\node at (10,10) {$q_1$};
\draw (4,-1) -- (11,10) -- (16,-1) ;
\node[red] at (11.5,7) {$m_2$};
\node at (11,5) {$q_2$};
\node[red] at (11,1) {$m_3$};
\draw (6,-1) -- (10,5) -- (14,-1) ;
\draw[red,dashed] (10,15) -- (11,10) -- (10,5) -- (10,-1);
\node[red] at (10,-2) {$\limsup \in \Qall$};
  \end{scope}

\begin{scope}[shift={(7,0)}, scale=.3]
\node at (5,16) {${q_0}$};
\node at (3,13) { $\rho_{q_0}$};
\draw[black] (0,-1) -- (5,15) -- (10,-1) ;
\draw[black] (2.5,-1) -- (2.5,7);
\draw[black] (7.5,-1) -- (7.5,7);
\node at (5,10) {$d_{q_0}$};
\node at (5,8) {$\mathbb{P}\geq \frac{1}{2}$};
\fill [opacity=0.1,blue] (2.5,-1) -- (2.5,7) -- (5,15) -- (7.5,7) -- (7.5,-1) -- cycle;
\draw[red,dashed] (5,15) -- (6,6) -- (4,3) -- (5,-1);
\node[red] at (5,-2) {$\limsup\in \Qone$};

\begin{scope}[shift={(2.5,1)}]
\begin{scope}[rotate around={-35:(0,2)}]
\draw (0,2) -- (0,1);
\draw (0,0.3) circle (0.7);
\node at (0,0.3) {$q_1$};
  \end{scope}
  \end{scope}
\begin{scope}[shift={(7.5,0)}]
\begin{scope}[rotate around={35:(0,2)}]
\draw (0,2) -- (0,1);
\draw (0,0.3) circle (0.7);
\node at (0,0.3) {$q_2$};
  \end{scope}
  \end{scope}
\begin{scope}[shift={(2.5,-1)}]
\begin{scope}[rotate around={-35:(0,2)}]
\draw (0,2) -- (0,1);
\draw (0,0.3) circle (0.7);
\node at (0,0.3) {$q_3$};
  \end{scope}
  \end{scope}

  \end{scope}


\begin{scope}[shift={(12,0.7)},scale=.14]

\begin{scope}[shift={(0,15)}]
\node at (5,9) {$d_{q_0}$};
\fill [opacity=0.1,blue] (2.5,-15) -- (2.5,7) -- (5,15) -- (7.5,7) -- (7.5,-15) -- cycle;
\begin{scope}[shift={(2.5,1)}]
\begin{scope}[rotate around={-45:(0,2)}]
\draw (0,2) -- (0,1);
\draw (0,0.3) circle (0.7);
  \end{scope}
  \end{scope}
  
\begin{scope}[shift={(7.5,-3)}]
\begin{scope}[rotate around={45:(0,2)}]
\draw (0,2) -- (0,1);
\draw (0,0.3) circle (0.7);
  \end{scope}
  \end{scope}
  
\begin{scope}[shift={(2.5,-5)}]
\begin{scope}[rotate around={-45:(0,2)}]
\draw (0,2) -- (0,1);
\draw (0,0.3) circle (0.7);
  \end{scope}
  \end{scope}

\begin{scope}[shift={(7.5,-12)}]
\begin{scope}[rotate around={45:(0,2)}]
\draw (0,2) -- (0,1);
\draw (0,0.3) circle (0.7);
  \end{scope}
  \end{scope}

  \end{scope}
  
  \begin{scope}[shift={(-10,7.5)}]
\begin{scope}[rotate around={-35:(0,2)}]
\node at (5,9) {$d_{q_1}$};
\fill [opacity=0.1,blue] (2.5,-1) -- (2.5,7) -- (5,15) -- (7.5,7) -- (7.5,-1) -- cycle;
  \end{scope}
  \end{scope}

  \begin{scope}[shift={(-8.3,-0.5)}]
\begin{scope}[rotate around={-25:(0,2)}]
\node at (5,9) {$d_{q_3}$};
\fill [opacity=0.1,blue] (2.5,-1) -- (2.5,7) -- (5,15) -- (7.5,7) -- (7.5,-1) -- cycle;
\begin{scope}[shift={(7.5,0)}]
\begin{scope}[rotate around={35:(0,2)}]
\draw (0,2) -- (0,1);
\draw (0,0.3) circle (0.7);
  \end{scope}
  \end{scope}
  \end{scope}
  \end{scope}

 \begin{scope}[shift={(12.0,-2.5)}]
\begin{scope}[rotate around={35:(0,2)}]
\node at (5,9) {$d_{q_2}$};
\fill [opacity=0.1,blue] (2.5,-1) -- (2.5,7) -- (5,15) -- (7.5,7) -- (7.5,-1) -- cycle;
\begin{scope}[shift={(2.5,1)}]
\begin{scope}[rotate around={-35:(0,2)}]
\draw (0,2) -- (0,1);
\draw (0,0.3) circle (0.7);
  \end{scope}
  \end{scope}
    \end{scope}
  \end{scope}
  
    \begin{scope}[shift={(5,-11.5)}]
\begin{scope}[rotate around={5:(0,2)}]
\fill [opacity=0.1,blue] (2.5,7) -- (5,15) -- (7.5,7)  -- cycle;
  \end{scope}
  \end{scope}

 \begin{scope}[shift={(9.5,-15)}]
\begin{scope}[rotate around={5:(0,2)}]
\fill [opacity=0.1,blue] (2.5,7) -- (5,15) -- (7.5,7)  -- cycle;
  \end{scope}
  \end{scope}

 \begin{scope}[shift={(-5,-18)}]
\begin{scope}[rotate around={5:(0,2)}]
\fill [opacity=0.1,blue] (3.5,10) -- (5,15) -- (6.5,10)  -- cycle;
  \end{scope}
  \end{scope}

\draw[red,dashed] (5,30) -- (4,13.3) -- (2,11.3) -- (-2,2) -- (-1,-7) ;
  \node[red] at (0,-9) {$\limsup\in \Qall$};

  \end{scope}

\end{tikzpicture}

\caption{The left picture illustrates how an accepting run is turned into a winning strategy for Automaton
in the jumping game, the two other pictures illustrate the converse transformation.}
\label{fig:jumping}
\end{figure}

\vspace{-0.3cm}

\begin{lemma}\label{lem:jumping-game-iff}
The automaton is nonempty if and only if  Automaton wins the jumping game.
\end{lemma}
\begin{proof}
The proof transforms an accepting run $\rho$ of the nonzero automaton into a winning strategy $\sigma$ of Automaton, and back, this is illustrated by Fig.~\ref{fig:jumping}.

Assume first that the nonzero automaton has an accepting run $\rho$.
Automaton can win the jumping game by
playing profiles of runs obtained as subtrees of $\rho$ rooted 
at deeper and deeper depths.
For a start, Automaton plays the profile $\Pi_0$ of $\rho$.
Then Pathfinder chooses some pair $(q_1,m_1)\in\Pi_0$,
by definition of profiles this corresponds to some non-root node $v_1$ of $\rho$ 
labelled by $q_1$, $m_1$ is the maximal state of the strict ancestors of $v_1$.
At each step $n>0$, Pathfinder chooses a pair $(q_n,m_n)\in\Pi_n$
corresponding to some node $v_{n+1}$ whose $v_{n}$ is a strict ancestor,
then Automaton plays the profile $\Pi_{n+1}$ of the subtree $\rho_n$ rooted in $v_n$.
Since $\rho$ is accepting then \emph{a fortiori} $\rho_n$
is $\set{\Qone,\Qnonzero}$-accepting.
The nodes $v_1,v_2,\ldots$ and their ancestors form a branch of $\rho$,
%
whose limsup is in $\Qall$ because $\rho$ is surely accepting.
This limsup is equal to $\limsup_n m_n$ thus Automaton wins the play.

Conversely, we use a positional winning strategy of Automaton to build an accepting run
of the nonzero automaton.
Denote $W$ the set of states winning for Automaton.
With every state $q$ in $W$ we associate the profile $\Pi_q$ chosen by the positional winning strategy of Automaton
and a $\set{\Qone,\Qnonzero}$-accepting run $\rho_q$ with profile $\Pi_q$.

We show the existence of a leaf-free subtree $d_q$ of $\rho_q$ 
such that:
\begin{enumerate}
\item[a)]
the set of branches of $d_q$ has probability $\geq \frac{1}{2}$,
\item[b)]
\emph{every} branch of $d_q$ has limsup in $\Qone$,
\item[c)]
for every node $v$ of $d_q$ with state in $\Qnonzero$,
	the set of branches of $d_q$ which visit $v$ and visit only $\Qnonzero$-labelled nodes below $v$
	has nonzero probability.
\end{enumerate}

Since $\rho_q$ is almost-surely accepting, then according to Lemma~\ref{lem:approx},
there is a subtree $d_q$ of $\rho_q$
whose set of branches has probability $\geq \frac{1}{2}$ and \emph{all of them} have limsup in $\Qone$
(while in the run $\rho_q$ there may be a non-empty set of branches with limsup in $\Qall \setminus \Qone$,
with probability zero).
Since we are only interested in branches of $d_q$, we can assume that $d_q$ is leaf-free.
This guarantees properties a) and b) but not c).
For every node $v$, define
$L_v$ the set of branches that visit $v$,
have limsup in $\Qone$ and visit only $\Qnonzero$-labelled nodes below $v$.
Since $\rho_q$ is $\set{\Qone,\Qnonzero}$-accepting,
for every node $v$ of $\rho_q$ with state in $\Qnonzero$,
$L_v$ has nonzero probability and according 
to Lemma~\ref{lem:approx} again, there exists a leaf-free subtree $d'_v$
whose every branch belongs to $L_v$.
We extend the definition domain of $d_q$  with $d'_v$.
This preserves properties a) and b) (because $d'_v$ itself has property b))
and guarantees property c).

Now we combine together the partial runs $(d_q)_{q\in W}$
in order to get an infinite graph.
Since $d_q$ is leaf-free, every node can have either both children in $d_q$ or only one child in $d_q$.
In case one child is missing,
we plug instead the partial run $d_r$, where $r$ is the state of the missing child in $\rho$.
This is well-defined because $r\in W$:
in a parity game, all states visited when playing a wining strategy are winning,
and after Automaton plays the profile $\Pi_q$ 
the next state of the game maybe any state appearing in $\rho_q$, including $r$.

The unravelling of this infinite graph, starting from the initial state,
is an accepting run of the automaton.
Each time a branch enters a subtree $d_q$,
there is probability $\geq \frac{1}{2}$ to stay in $d_q$ forever.
Thus almost every branch of the unravelling
eventually stays in one of the subtrees $(d_q)_{q\in W}$,
thus has limsup in $\Qone$ according to property b).
As a consequence the unravelling is almost-surely accepting.

Still, with probability $0$, some branches switch infinitely often from a subtree to another. 
Such a branch enters the $n$-th subtree $d_n$ in its root state $q_n$,
follow a path in $d_n$ with maximal state $m_{n+1}$
and exits $d_n$ to enter $d_{n+1}$ in state $q_{n+1}$.
Since $d_n$ is a subtree of $\rho_n$, then $(q_{n+1},m_{n+1})$ is in the profile $\Pi_{q_n}$
and $q_0\stackrel {q_0} \to \Pi_{q_0} \stackrel {m_1} \to q_1 \to \Pi_{q_1} \stackrel {m_2}\ldots$
is a play consistent with the winning strategy of Automaton.
Since the strategy of Automaton is winning then
$\limsup_n m_{n+1}\in\Qall$.
Hence the unravelling is surely accepting.

Moreover the unravelling is nonzero accepting as well according to c).
\end{proof}

\begin{lemma}\label{lem:jumping-game-decide}
	Deciding whether  Automaton wins the jumping game is in \npconp.
\end{lemma}
\begin{proof}
The jumping game is a parity game thus the winner of the jumping game can be found by guessing a positional strategy for either Automaton or Pathfinder and checking that this strategy is winning.
However, since there are exponentially many profiles,
this algorithm is in {\sc nexptime} $\cap$ co-{\sc nexptime} rather than in \npconp.

To overcome this difficulty, we use \emph{winning witnesses} which are condensed versions of winning positional strategies of Automaton and Pathfinder.
A winning witness is a pair $(W,s)$
with $W\subseteq Q$ and $s : W \to 2^{W \times W}$.

Under extra-conditions, some of these witnesses are proofs that Automaton or Pathfinder is the winner of the jumping game.
A sequence $m_0, m_1, \ldots \in W^*$ is \emph{generated} by
$(W,s)$ if there exists $q_0,q_1,\ldots \in W^*$
such that $\forall n, (q_{n+1},m_{n+1})\in \sigma(q_n)$.
For every state $q$, denote $\mathcal{R}_q$
the set of profiles 
of $\set{\Qone,\Qnonzero}$-accepting runs with root state $q$.
Then $(W,s)$ is a \emph{winning witness for Automaton} if:
\begin{itemize}
\item[$(\alpha$)]
Every sequence generated by $(W,s)$ has its limsup 
 in $\Qall$
 and $\forall q \in W, s(q)\neq \emptyset$.
\item[$(\beta$)]
For every state $q\in W$
there exists a profile $\Pi\in\mathcal{R}_q$
such that $\Pi \subseteq s(q)$.
\end{itemize}
And $(W,t)$ is a \emph{winning witness for Pathfinder} if:
\begin{itemize}
\item[$(\gamma$)]
No sequence generated by $(W,t)$ has its $\limsup$ in $\Qall$.
\item[$(\delta$)]
For every state $q\in W$ and every profile $\Pi\in\mathcal{R}_q$, $t(q)\cap \Pi \neq \emptyset$. 
\end{itemize}

\begin{lemma}
Let $(W_A,W_P)$ be the partition of $Q$
between states winning for Automaton or Pathfinder
in the jumping game.
Then there exists winning witnesses $(W_A,s)$ and
$(W_P,t)$ for Automaton and Pathfinder.
\end{lemma}
\begin{proof}
We start with the direct implication.
In a parity game, a play consistent with a winning strategy never exits the set of winning vertices, thus Automaton and Pathfinder have positional winning strategies
 $\sigma: W_A \to 2^{W_A\times W_A}$
and $\tau : 2^{W_P\times W_P} \to W_P$.
Then $(W_A,\sigma)$ is a winning witness for Automaton:
property $(\alpha)$ holds because $\sigma$ is winning and property $(\beta)$ holds because, by definition of the jumping game, $\forall q\in W_A,\sigma(q)\in\mathcal{R}_q$.
The winning witness $(W_P,t)$ for Pathfinder is defined by
\[
t(q)=\{ (q',m) \mid \exists \Pi \in \mathcal{R}_q,
 \Pi\stackrel m \to q' \text{ is consistent with $\tau$}\}\enspace.
\]
Then $(W_P,t)$ is a winning witness for Pathfinder:
property $(\gamma)$ holds because $\tau$ is winning and property $(\delta)$ holds by definition of $t$.

We show the converse implication.
Let $(W_A,s)$ be a winning witness for Automaton. Then according to $(\beta)$,
for every $q\in W_A$ there exists 
a profile $\sigma(q)\in  \mathcal{R}_q$ such that $\sigma(q)\subseteq s(q)$.
Then $\sigma$ is a strategy in the jumping game
and according to $(\alpha)$ the strategy $\sigma$
is winning on $W_A$.
From  a winning witness $(W_P,t)$ for Pathfinder
we extract a positional strategy $\tau$ winning on $W_P$.
Let $q\in W_P$ 
and $\Pi$ a profile in $\mathcal{R}_q$.
According to $(\delta)$
there exists 
$(q,m) \in \Pi \cap t(q)$,
and $\tau$ plays the move
$\Pi \stackrel m \to q$ .
Then $\tau$ is winning
according to $(\gamma)$.
\end{proof}

Now we show how to check in polynomial time whether a pair $(W,s)$ is a winning witness for Automaton or Pathfinder.
Checking properties $(\alpha)$ or $(\gamma)$ consists in solving a one-player parity game which can be done in polynomial time.

To check properties $(\beta)$ or $(\delta)$,
we modify the automaton to store in its state space the maximal state of the strict ancestors of the current node.
The new state space is $Q\times (\{\bot\}\cup Q)$
and for every $m \in \{\bot\}\cup Q$,
every transition $q \to_a (q_0,q_1)$ in the original
automaton gives rise to a transition
$(q,m) \to_a ((q_0,m'),(q_1,m'))$
in the modified automaton
with $m'=q$ if $m=\bot$ and $m' = \max \{ m,q \}$ otherwise.
This extra component has no incidence on the acceptance condition. This transformation guarantees
that
for every state $q\in Q$ and every subset $\Pi \subseteq Q\times Q$,
\begin{itemize}
\item[($\star$)]
$\Pi \in \mathcal{R}_q$ if and only if
the modified automaton has a
$\set{\Qone,\Qnonzero}$-accepting run $\rho$ with root state $(q,\bot)$  and $\Pi$ is the set of states appearing on non-root nodes of $\rho$.
\end{itemize}
According to ($\star$), property $(\beta)$ is equivalent to checking that
for every $q\in W$,
the modified automaton restricted to states in 
$\{(q,\bot)\}\cup s(q)$ has a $\set{\Qone,\Qnonzero}$-accepting run,
which can be done in polynomial time according to 
Theorem~\ref{theo:probanp}.

And according to ($\star$), property $(\delta)$ is equivalent to checking that
for every $q\in W$,
the modified automaton restricted to states in 
$\{(q,\bot)\}\cup Q\times Q \setminus s(q)$ has no $\set{\Qone,\Qnonzero}$-accepting run,
which can be done in polynomial time according to 
Theorem~\ref{theo:probanp}.
\end{proof}

\paragraph*{Example: the everywhere positive language}
A tree $t$ on the alphabet $\set{a,b}$ is  \emph{everywhere positive} if
for
every node $v$,
\begin{enumerate}
\item there is positive probability to see only the letter $t(v)$ below $v$,
\item there is positive probability to see finitely many times the letter $t(v)$ below $v$.
\end{enumerate}

This language is non-empty and contains no regular tree.
The language of everywhere positive trees with root state $a$ is recognized by a nonzero automaton
with six states 
\[
\set{s_b < s_a  < n_b < n_a  < f_b < f_a}\enspace.
\]
On a node labelled by letter $a$,
the automaton can perform a transition
from any of the three states
$\set{s_b,n_b,f_a}$,
meaning intuitively "searching for $b$",
"not searching for $b$" and "just found $a$".
From these states the automaton can choose any pair of successor states which intersects $\set{s_b,f_b}$.
Transitions
on letter $b$ are symmetrical.
The acceptance condition is:
 \begin{align*}
 	\Qall = \set{n_a,n_b,f_a,f_b} \qquad \Qone = \Qall \qquad \Qnonzero = \{na,sa,nb,sb\}\enspace.
 \end{align*}
Due to space constraints,
we can not provide a full description of the jumping game
(see the appendix for more details).
Automaton can win by playing only the moves
$
s_a/n_a\to\set{(f_a,f_a),(n_b,f_a), (s_b,f_a), (n_a,n_a), (s_a,n_a)} 
$
and
$
f_a\to\set{(n_b,f_a), (s_b,f_a)} 
$
and their symmetric counterparts from states $\set{s_b,n_b,f_b}$.
This forces Pathfinder to take only edges labelled by
the states $\set{f_a,n_a,f_b,n_b}$.
These states dominate the states $\set{s_a,s_b}$
thus the limsup of the corresponding plays is in $\Qall$
and
this is a winning strategy for Automaton.

\section*{Conclusion}
We have shown that the emptiness problem for zero and nonzero automata is decidable and in \npconp.
As a consequence, the satisfiability for the logic {\sc MSO +} $\zero$ from~\cite{thinzero} is decidable (in non-elementary time), when $\zero$ is the unary predicate that checks a set of branches has probability $0$.

As shown by Stockmeyer, the satisfiability problem for first-order logic on finite words cannot be solved in elementary time. Therefore any translation from a logic stronger than first-order logic on finite words (such as  {\sc tmso}+$\zero$ on infinite trees) to an automaton model with elementary emptiness (such as nonzero automata) is necessarily nonelementary. This does not make the relatively low \npconp\ complexity of nonzero automata any less interesting. One can imagine other logics than TMSO+zero, either less expressive or maybe even equally expressive but less succint, which will have a relatively low complexity by virtue of a translation into nonzero automata. One natural direction is the study of temporal logics.

%
\subparagraph*{Acknowledgments}

We thank Paulin Fournier, Henryk Michalewski and Matteo Mio for  helpful discussions. 

\bibliography{bib}

\appendix
\newpage
\section*{Appendix}

\section*{Proof of Lemma~\ref{lem:strong}}

\begin{proof}
Clearly every strongly zero accepting run is also zero accepting.

Conversely,
assume a run $r$ is zero accepting, then we show it is strongly zero accepting. Let $v$ be a node labelled by a seed state.
Among all descendant nodes $z$ of $v$, including $v$ itself, 
such that the path from $v$ to $z$ is seed-consistent and $z$ is labelled by a seed state,
 choose any $z$ such that the seed state labelling $z$ is minimal.

For every node $w$ let
$Z_w$ denote the set of branches which visit $w$ and afterwards see only states $\leq  r(w)$ and have limsup $ r(w)$.

We first show that there exists a strict descendant $w$ of $z$ such that 
\begin{itemize}
\item[a)]
$ r(w)\in \Qnonzero$,
\item[b)]
the path from $z$ to $w$ is labelled by states $\leq  r(z)$ and
\item[c)]
$Z_{w}$ has nonzero probability. 
\end{itemize}
Since $r$ is zero accepting and $z$ is labelled by a seed state, there is at least one descendant node $w'$ of $z$, labelled by a state in $\Qnonzero$, such that the path from $z$ to $w'$ is labelled by states $\leq  r(z)$ and $Z_{w'}$ has nonzero probability.
If $w'$ is a strict descendant of $z$ then we set $w=w'$.
Otherwise we choose $w$ as a strict descendant of $w'$, as follows.
Denote $W$ the set of strict descendants of $w'$ which are labelled by $ r(w')$ and the path from $w'$ to $w$ is labelled by states $\leq  r(z)$.
Then $Z_{w'} = \bigcup_{w\in W} Z_{w}$ thus by 
$\sigma$-additivity there exists
a strict descendant $w$ of $w'$ such that  $Z_{w}$ also has non-zero probability.

To establish that the strongly zero accepting condition is satisfied for $v$, we choose a witness $w$ satisfying properties a) b) and c) and we prove two other properties of $w$:
\begin{itemize}
\item[d)] the path from $v$ to $w$ is seed-consistent,
\item[e)] the only seed state that may be visited below $w$ by a branch in $Z_w$ is $ r(w)$ itself.
\end{itemize}

Property d) holds because both paths from $v$ to $z$ and from $z$ to $w$ are seed-consistent and the concatenation of two seed-consistent pathes on a $\Qseed$-labelled node is itself a seed-consistent path. The path from $v$ to $z$ is seed-consistent by choice of $z$. By hypothesis the path from $z$ to $w$ is labelled by states $\leq  r(z)$ and by minimality of $ r(z)$ it does not meet any other seed state than $ r(z)$ thus it is seed consistent.

Property e) holds for a similar reason: if a branch in $Z_w$ visit a descendant $z'$ of $w$ such that $ r(z')\in \Qseed$
then by definition of $Z_w$, $ r(z')\leq  r(w)$.
Since $ r(w)\leq  r(z)$, the path from $z$ to $z'$ is labelled by states $\leq  r(z)$ and the minimality of $ r(z)$ it implies $ r(z)\leq  r(z')$ thus finally $ r(z') =  r(w) =  r(z)$.
\end{proof}

\section*{Proof of Lemma~\ref{lem:projnonzero}}

\begin{proof}
By hypothesis $\Pi_1:(R,\preceq)\to (Q,\leq)$ is monotonic,
thus if $b$ is a branch of the infinite binary tree then its limsup in $r$ is the projection of its limsup in $d$.

Since $\Qnonzero \subseteq \Qone \subseteq \Qall$ 
then the projection of $\Rone$ is $\Qone$ and the projection of $\Rall$ is $\Qall$ thus $r$ is both almost-surely and surely accepting.

We show that $r$ is zero accepting.
Let $v$ a node such that $r(v)$ is a seed state.

For a start, we show that there is a node $w$ below $v$
such that the path from $v$ to $w$ is seed-consistent in $r$
(thus in particular $r(w) \leq r(v)$) and $d(w)$ is the subtree-guessing state $(r(w),r(w),*)$.
There are three cases, depending whether $d(v)$ is a subtree-guessing, path-finding or normal state. If $d(v)$ is a subtree guessing state then according to the definition of $R$,
since $r(v)\in\Qseed$ then 
$d(v)=(r(v),r(v),*)$ and we set $w=v$.
If $d(v)$ is a path-finding state
then by design 
the automaton follows in either direction a  path seed-consistent in $r$ as long as it does not enter a subtree-guessing state $(r(w),r(w),*)$. Since there is no path-finding state in $\Rall$, for sure the automaton eventually enters such a state,
otherwise $d$ would not be accepting.
If $d(v)$ is a normal state then according to the transition table either the left or right child $w'\in\{v0,v1\}$ of $v$ is in the path-finding state $(r(w'),r(v))$ or the subtree-guessing state $(r(w'),r(w'),*)$. In both cases $r(w')\leq r(v)$. In the subtree-guessing case we set $w=w'$ and we are done. In the pathfinding case, from $w'$ the automaton follows a path 
seed-consistent in $r$ until it eventually enters the subtree-guessing state $(r(w),r(w),*)$. By design of the transition table all states on the path from $w'$ to $w$ are $\leq r(v)$ thus the path from $v$ to $w$ is seed-consistent in $r$.
 
Since $(r(w),r(w),*)\in\Rnonzero$, the nonzero condition ensures that there is nonzero probability to continue the run $r$ below $w$ in the set of states $\Rnonzero$.
According to the transition table, in this case the states below $w$ are labelled by $\{ q\in Q \mid q\leq r(w)\} \times \{r(w)\}\times \{*\}$. Since $d$ is almost-surely accepting then by definition of $\Rone$, almost-surely the limsup of such a path is $(r(w),r(w),*)$.
Since $r(w)\in\Qnonzero$ then the nonzero condition holds in $v$, with witness $w$.
\end{proof}

\section*{Extended example: the everywhere positive language}
A tree $t$ on the alphabet $\set{a,b}$ is  \emph{everywhere positive} if
for
every node $v$,
\begin{enumerate}
\item there is positive probability to see only the letter $t(v)$ below $v$,
\item there is positive probability to see finitely many times the letter $t(v)$ below $v$.
\end{enumerate}

This language is non-empty and contains no regular tree.
The language of everywhere positive trees with root state $a$ is recognized by a nonzero automaton
with six states 
\[
\set{s_b < s_a  < n_b < n_a  < f_b < f_a}\enspace.
\]
On a node labelled by letter $a$,
the automaton can perform a transition
from any of the three states
$\set{s_b,n_b,f_a}$,
meaning intuitively "searching for $b$",
"not searching for $b$" and "just found $a$".
From these states the automaton can choose any pair of successor states which intersects $\set{s_b,f_b}$.
Transitions
on letter $b$ are symmetrical.
The acceptance condition is:
 \begin{align*}
 	\Qall = \set{n_a,n_b,f_a,f_b} \qquad \Qone = \Qall \qquad \Qnonzero = \{na,sa,nb,sb\}\enspace.
 \end{align*}

We do not provide a full description of the jumping game
but we provide a few examples of moves available to player Automaton,
as well as a positional winning strategy for player Automaton.

Among the simplest moves of Automaton in the jumping game are the two moves
\begin{align*}
&n_b \to \{(n_b,n_b)(s_b,n_b)\}
\\
&s_b \to \{(n_b,n_b)(s_b,n_b)\}\enspace.
\end{align*}
These moves are legal because they are the profiles of
the following $\set{\Qone,\Qnonzero}$-accepting runs.
Both runs are on the tree whose all nodes have letter $a$
and everywhere in the tree the automaton applies the same 
two transitions
$n_b \to_b (n_b,s_b)$
and
$s_b \to_b (n_b,s_b)$.
In other words, the automaton always looks for a letter $b$ in the right direction
(state $s_b$),
and does not look for $b$ in the left direction (state $n_b$).
Since the tree has no $b$  then the quest for a letter $b$ is hopeless,
and on are branches of the run that ultimately always turn right
(i.e. branches in $\{0,1\}^*1^\omega$),
the automaton ultimately stays in state $s_b$ and the branch has limsup $s_b$,
which is neither in $\Qall$ nor in $\Qone$.
But such branches happen with probability zero:
almost-every branch makes infinitely many turns left and right
and has limsup $n_b$,
thus the run is almost-surely accepting:
This run is nonzero-accepting as well because every node labelled by $\Qnonzero$
 has all its descendants labelled by $\Qnonzero$.
 
 Yet legal, these two moves are not good options for Automaton in the jumping game
 because then Pathfinder can generate the play
 \[
 s_b \stackrel {s_b} \to \set{(n_b,n_b)(s_b,n_b)} \stackrel {n_b} \to s_b \stackrel {s_b} \to \set{(n_b,n_b)(s_b,n_b)}\stackrel {s_b} \to s_b \stackrel {s_b} \to \ldots
 \]
 which has limsup $n_b=\max\{s_b,n_b\}$ and is losing for Automaton since $n_b\not\in\Qall$.

\medskip

Automaton should use more elaborate moves in order to win the jumping game,
in particular the three moves
\begin{align}
\label{eq:moves}
&s_a/n_a\to\set{(f_a,f_a),(n_b,f_a), (s_b,f_a), (n_a,n_a), (s_a,n_a)}\\
\label{eq:moves2}
&f_a\to\set{(n_b,f_a), (s_b,f_a)} 
\end{align}
are interesting.
Before explaining which these are legal moves,
remark that these three moves 
and their symmetric counterparts from states $\set{s_b,n_b,f_b}$
ensure the victory to Automaton,
because they force
Pathfinder to take edges labelled by
the states $\set{f_a,n_a,f_b,n_b}$.
These four states dominate the states $\set{s_a,s_b}$
and belong to $\Qall$
thus the limsup of the corresponding plays are in $\Qall$,
which ensures a win to Automaton.

We show that~\eqref{eq:moves} and~\eqref{eq:moves2} 
are legal moves for Automaton in the jumping game,
by providing  positional runs of the extended automaton
which generate the profiles
$\set{(n_b,f_a), (s_b,f_a)}$
and
$\set{(f_a,f_a),(n_b,f_a), (s_b,f_a), (n_a,n_a), (s_a,n_a)}$.
%
%

We start with a brief description of the extended automaton.
To save  space, we write $s_*$ for the pair $\set{s_a,s_b}$
and use a similar convention for $n_*$ and $f_*$ as well.
With this convention, the states 
are
\[
\set{s_*,n_*,f_*}\times\set{\bot, s_*,n_*,f_*}\enspace.
\]
On the  first component,
the transitions of the extended automaton are identical to the
transitions of the original automaton.
The second component is used to store 
the largest state seen so far.
It is initialized to $\bot$
and then updated with the maximum of itself and the origin state
of the transition.

We give three examples of transitions of the extended automaton
\begin{itemize}
\item
The automaton starts the computation looking for an $a$ and keeps looking for an $a$ on the left direction:
\[
(s_a,\bot)\to_b (s_a,s_a)(n_a,s_a)\enspace.
\]
\item
The automaton is not looking for an $a$
but it finds an $a$ in the left child and keeps looking for an $a$
in the right direction:
\[
(n_a,s_a)\to_b (f_a,n_a)(s_a,n_a)\enspace.
\]
\item
The automaton has already found $b$ in the past,
it is right now looking for an $a$,
and finds one $a$ in both direction:
\[
(s_a,f_b)\to_b (f_a,f_b)(f_a,f_b)\enspace.
\]
\end{itemize}
This last transition is a killer for the nonzero condition,
because $s_a\in \Qnonzero$ but $f_a\not \in \Qnonzero$.
Using this transition falsifies
the condition "there is positive probability to see only the letter $b$ below $v$"
is not satisfied.
Actually this transition could be removed from the set of transitions without changing 
the set of accepting runs.

To prove that the move
$f_a\to\set{(n_b,f_a), (s_b,f_a)}$
is valid, we consider the run on a tree whose all nodes are labelled by $a$.
The extended automaton first find an $a$ in the root, in state $(f_a,\bot)$
and then looks hopelessly for a $b$ in the right direction
using the transitions
\begin{align*}
& (f_a,\bot) \to_a (n_b,f_a)(s_b,f_a)\\
& (nb/s_b,f_a) \to_a (n_b,f_a)(s_b,f_a)\enspace.
\end{align*}
 This run is almost-surely
 accepting
  because
every branch which takes infinitely many turns left
has limsup $(n_b,f_a)$,
and this is almost-every branch.
 This run is nonzero-accepting because every node labelled by $\Qnonzero$
 has all its descendants labelled by $\Qnonzero$.
 
 \medskip
 
To prove that the move
 $
s_a\to\set{(f_a,f_a),(n_b,f_a), (s_b,f_a), (n_a,n_a), (s_a,n_a)} 
$
is legal, consider 
a tree whose root is labelled by $b$,
all the nodes in the left subtree are labelled by $b$ as well
while all the nodes in the left subtree are labelled by $a$.
The extended automaton starts on state $(s_a,\bot)$ in the root.
In the right subtree the automaton finds $b$ on the right child of the root
(i.e. node $1$)
and then looks hopelessly for $a$ in the right direction
using transitions
$n_a/s_a \to_a (n_a,s_a)$
(dual to the previous case
$f_a\to\set{(n_b,f_a), (s_b,f_a)}$).
In the left subtree
the automaton 
looks hopelessly for $b$ in the right direction
using transitions
$n_b/s_b \to_b (n_b,s_b)$.

The transitions of this positional run are
\begin{align*}
& (s_a,\bot) \to_a (f_a,s_a)(n_a,s_a) \text{ (used once in the root) }\\
& (n_a,s_a) \to_a (n_a,n_a)(s_a,n_a)  \text{ (used once in the right subtree) }\\
& (n_a/s_a,n_a) \to_a (n_a,n_a)(s_a,n_a)  \text{ (used $\infty$ often in the right subtree) }\\
& (f_a,s_a) \to_a (n_b,f_a)(s_b,f_a)  \text{ (used once in the left subtree) }\\
& (n_b/s_b,f_a) \to_a (n_b,f_a)(s_b,f_a)  \text{ (used $\infty$ often in the left subtree) }\enspace.
\end{align*}
This run is almost-surely accepting for the same reasons than in the previous case.
It is nonzero accepting because from the root node,
whose state $(s_a,\bot)$ is in $\Qnonzero$
there is probability $\frac{1}{2}$ to continue in the right subtree
where all states are in $\Qnonzero$.
And every non-root node labelled by $\Qnonzero$
 has all its descendants labelled by $\Qnonzero$.

\medskip

The positional run 
for the move
$
n_a\to\set{(f_a,f_a),(n_b,f_a), (s_b,f_a), (n_a,n_a), (s_a,n_a)} 
$
is almost the same than for the move
$
s_a\to\set{(f_a,f_a),(n_b,f_a), (s_b,f_a), (n_a,n_a), (s_a,n_a)} 
$
except the root has state $(n_a,\bot)$.
The transitions of this positional run are
\begin{align*}
& (n_a,\bot) \to_a (f_a,n_a)(n_a,n_a) \text{ (used once in the root) }\\
& (f_a,n_a) \to_a (n_b,f_a)(s_b,f_a)  \text{ (used once in the left subtree) }\\
& (n_b/s_b,f_a) \to_a (n_b,f_a)(s_b,f_a)  \text{ (used $\infty$ often in the left subtree) }\\
& (n_a/s_a,n_a) \to_a (n_a,n_a)(s_a,n_a)  \text{ (used $\infty$ often in the right subtree) }\enspace.
\end{align*}
It is $\set{\Qone,\Qnonzero}$-accepting for the same reasons than in the previous case.

\end{document}